\documentclass[conference]{IEEEtran}

\usepackage{cite}
\usepackage{amsfonts,amssymb,amsmath,amsthm,graphicx}
\usepackage{url}
\usepackage{dsfont}
\usepackage{bbm}

\theoremstyle{definition}
\newtheorem{theorem}{Theorem}
\newtheorem{proposition}[theorem]{Proposition}

\newtheorem{definition}[theorem]{Definition}

\newcommand*{\eps}{\varepsilon}
\newcommand*{\trace}{\mathrm{tr}}
\newcommand*{\ket}[1]{| #1 \rangle}
\newcommand*{\bra}[1]{\langle #1 |}

\newcommand*{\braket}[2]{\langle #1 | #2 \rangle}
\newcommand*{\proj}[1]{\ket{#1}\bra{#1}}

\newcommand*{\nC}{\mathbb{C}}

\newcommand*{\cB}{\mathcal{B}}

\newcommand*{\cH}{\mathcal{H}}
\newcommand*{\cI}{\mathcal{I}}

\newcommand*{\cP}{\mathcal{P}}
\newcommand*{\cS}{\mathcal{S}}

\newcommand*{\id}{\mathbbm{1}}

\newcommand{\rhoh}{\hat{\rho}}

\usepackage{color}

\begin{document}

\title{Variations on Classical and Quantum Extractors}

\author{\IEEEauthorblockN{Mario Berta}
\IEEEauthorblockA{IQIM Caltech, Pasadena\\
Institute for Theoretical Physics, ETH Zurich}
\and
\IEEEauthorblockN{Omar Fawzi and Volkher B.~Scholz}
\IEEEauthorblockA{Institute for Theoretical Physics\\
ETH Zurich}
\and
\IEEEauthorblockN{Oleg Szehr}
\IEEEauthorblockA{Centre for Mathematics\\
Technical University Munich}}

\maketitle


\begin{abstract}
\boldmath
Many constructions of randomness extractors are known to work in the presence of quantum side information, but there also exist extractors which do not [Gavinsky {\it et al.}, STOC'07]. Here we find that spectral extractors $\psi$ with a bound on the second largest eigenvalue $\lambda_{2}(\psi^{\dagger}\circ\psi)$ are quantum-proof. We then discuss fully quantum extractors and call constructions that also work in the presence of quantum correlations decoupling. As in the classical case we show that spectral extractors are decoupling. The drawback of classical and quantum spectral extractors is that they always have a long seed, whereas there exist classical extractors with exponentially smaller seed size. For the quantum case, we show that there exists an extractor with extremely short seed size $d=O(\log(1/\epsilon))$, where $\eps>0$ denotes the quality of the randomness. In contrast to the classical case this is independent of the input size and min-entropy and matches the simple lower bound $d\geq\log(1/\epsilon)$.
\end{abstract}

\IEEEpeerreviewmaketitle


\section{Introduction}\label{sec:intro}

Randomness is fundamental for many applications in computation, cryptography and information theory, and the goal of randomness extraction is to convert sources of biased and correlated bits to almost-uniform bits (see, e.g.,~\cite{Vadhan11}). A min-entropy extractor takes an input string $N$ from a weakly random source and applies a function $f$ together with a string $D$ of perfect randomness (called seed) to yield an output string $M=f(N,D)$ which is then supposed to be $\eps$-close to uniform provided that the min-entropy $H_{\min}(N)=-\log p_{\mathrm{guess}}(N)$ is large enough. However, for some applications we also want that the extractor works if the input source is correlated to another system $R$. That is, the output should be uniform and independent of $R$ provided that the conditional min-entropy of the source $H_{\min}(N|R)=-\log p_{\mathrm{guess}}(N|R)$ is large enough. Extractor constructions that also work if $R$ is quantum are called quantum-proof. Such extractors are crucial in classical and quantum cryptography (see, e.g.,~\cite{Renner05}) as well as for proving quantum coding theorems (see, e.g.,~\cite{Dupuis12,Berta13}). In the first part of this paper, we briefly discuss what is known about extractors (Section~\ref{sec:minentropy}) and then show that spectral extractors $\psi$ with a bound on the second largest eigenvalue $\lambda_{2}(\psi^{\dagger}\circ\psi)$ are quantum proof (Section~\ref{sec:quantumproof}). In the second part of this paper (Sections~\ref{sec:quantum_min} and~\ref{sec:shortseeded}), we consider fully quantum min-entropy extractors that output a quantum state that is $\eps$-close to maximally mixed from a quantum source $\rho_{N}$ provided that the min-entropy $H_{\min}(N)_{\rho}=-\log\lambda_{1}(\rho_{N})$ is large enough ($\lambda_{1}(\cdot)$ denotes the largest eigenvalue). We then discuss an extension of this when the input source is correlated to another quantum system $R$, and we require the output to be to uniform and independent of $R$ provided that the quantum conditional min-entropy $H_{\min}(N|R)_{\rho}=-\log F(N,R)_{\rho}$ is large enough ($F(N,R)_{\rho}$ denotes the maximal achievable singlet fraction~\cite{Koenig09}). We call such extractors decoupling and show that spectral extractors are decoupling. Finally, we show that spectral extractors have the drawback of a long seed, but also give a direct extractor construction with extremely short seed. We end by stating some open questions (Section~\ref{sec:discussion}).

We use the following notation. The labels $N,M,D$ are used to specify the subsystem as well as the domain of the classical system, and classical states on $N$, i.e., probability distributions on $N$, are denoted by $P_{N}\in\ell^1(N)$. The normalized uniform distribution on $N$ is denoted by $u_{N}$, and the set of distributions on $N$ is denoted by $\ell(N)$. Quantum systems are represented by their (finite-dimensional) Hilbert spaces $\cH_{N},\cH_{M},\cH_{R}$, and states on $\cH_{N}$, i.e., non-negative trace-one operators on $\cH_{N}$ are denoted by $\rho_{N}\in\cS(\cH_{N})$. We denote the set of linear operators on $\cH_{N}$ by $\cP(\cH_{N})$. For $\alpha\geq1$ and $\sigma\in\cP(\cH)$ with $\sigma\geq0$ we have the $\sigma$-weighted $\alpha$-norms
\begin{align}
\|\cdot\|_{\alpha,\sigma}=\Big(\trace\big[|\sigma^{1/2\alpha}(\cdot)\sigma^{1/2\alpha}|^{\alpha}\big]\Big)^{1/\alpha}\ ,
\end{align}
and for $\alpha=2$ the norm is induced by the $\sigma$-weighted Hilbert-Schmidt inner product
\begin{align}
\langle\cdot|*\rangle_{\sigma}=\trace\Big[\big(\sigma^{1/4}(\cdot)\sigma^{1/4}\big)^{\dagger}\big(\sigma^{1/4}(*)\sigma^{1/4}\big)\Big]\ .
\end{align}
The $\sigma$-weighted $\alpha$-norms satisfy the H\"{o}lder inequalities
\begin{align}
\langle\cdot|*\rangle_{\sigma}\leq\|\cdot\|_{p,\sigma}\cdot\|*\|_{q,\sigma}
\end{align}
for $1/p+1/q=1$~\cite{Olkiewicz99}. For $\rho_{NR}\in\cS(\cH_{NR})$ the quantum conditional min-entropy is defined as~\cite{Koenig09}
\begin{align}
H_{\min}(N|R)_{\rho}=-\log\max_{\substack{E_{NR}\geq0\\\trace_{N}[E_{NR}]=\id_{R}}}\trace\left[\rho_{NR}E_{NR}\right]
\end{align}
and the classical-quantum version simplifies to $H_{\min}(N|R)_{\rho}=-\log p_{\mathrm{guess}}(N|R)_{\rho}$, where $p_{\mathrm{guess}}(N|R)_{\rho}$ is the maximal probability of decoding $N$ from measurements on $R$. For $\rho_{NR}\in\cS(\cH_{NR})$ the quantum conditional R\'{e}nyi-two entropy is defined as $H_{2}(N|R)_{\rho}=-\log\|\tilde{\rho}_{NR}\|_{2}^{2}$ with $\tilde{\rho}_{NR}=(\id_{N}\otimes\rho_{R}^{-1/4})\rho_{NR}(\id_{N}\otimes\rho_{R}^{-1/4})$, and we have~\cite{Berta13_3}
\begin{align}\label{eq:h2vshmin}
H_{\min}(N|R)_{\rho}\leq H_{2}(N|R)_{\rho}\ .
\end{align}


\section{Classical Min-Entropy Extractors}\label{sec:minentropy}

The definition of a (strong) extractor is due to Nisan and Zuckerman~\cite{Nissan96}.

\begin{definition}\label{def:cmin_ext}
Let $M\subset N$, $k\in[0,\log|N|]$, and $\eps>0$. A $(k,\eps)$ extractor is a set of functions $\{f_{1},\ldots,f_{|D|}\}$ from $N$ to $M$ such that for all $P_{N}\in\ell^1(N)$ with $H_{\min}(N)_{P}\geq k$,
\begin{align}\label{eq:minextractor_class}
\big\|\frac{1}{|D|}\cdot\sum_{i=1}^{|D|}P_{f_{i}(N)}\otimes\proj{i}_{D}-u_{M}\otimes u_{D}\big\|_{1}\leq\eps\ .
\end{align}
The quantity $n=\log|N|$ is called the input size, $m=\log|M|$ the output size, and $d=\log|D|$ the seed size.\footnote{For a weak $(k,\eps)$ extractor~\eqref{eq:minextractor_class} is replaced with $\big\|\frac{1}{|D|}\cdot\sum_{i=1}^{|D|}P_{f_{i}(N)}-u_{M}\big\|_{1}\leq\eps$ (i.e., the seed system $D$ is not necessarily part of the output).}
\end{definition}

An extractor is called permutation based if all the functions $f_{i}:N\to M$ have the form $f_{i}(\cdot)=\pi_{i}(\cdot)\big|_{M}$ with $\pi_{i}\in S_{|N|}$, the symmetric group on $\{1,2,\ldots,|N|\}$. It is instructive to consider extractors with domain and range consisting of bit strings, that is, $N=\{0,1\}^{n}$, $M=\{0,1\}^{m}$, $D=\{0,1\}^{d}$. Typically we are given fixed $n$, $k$, and $\eps$, and we want to maximize the output length $m$ and minimize the seed length $d$. Radhakrishnan and Ta-Shma gave an ultimate limit on $m$ and $d$: every $(k,\eps)$ extractor necessarily has
\begin{align}\label{eq:tashma_bounds}
&m\leq k-2\log\frac{1}{\eps}+O(1)\\
&d\geq\log(n-k)+2\log\frac{1}{\eps}-O(1)\ .
\end{align}
It turns out that a probabilistic construction using random functions achieves these bounds up to constants: there exists a $(k,\eps)$ extractor with~\cite{Radhakrishnan00,Sipser88}
\begin{align}
&m=k-2\log\frac{1}{\eps}-O(1)\\
&d=\log(n-k)+2\log\frac{1}{\eps}+O(1)\ .
\end{align}
However, for applications we usually want explicit extractors and starting with Trevisan's breakthrough result~\cite{Trevisan99} there are now many constructions that almost achieve the bounds above (see~\cite{Vadhan11} and references therein). Here, we study the question if extractors also work in the presence of quantum information.


\section{Quantum-Proof Extractors}\label{sec:quantumproof}


\begin{definition}
A $(k,\eps)$ extractor $\{f_{1},\ldots,f_{|D|}\}$ is quantum-proof if for all classical-quantum states $\rho_{NR}\in\cS(\cH_{NR})$ with $H_{\min}(N|R)_{\rho}\geq k$,
\begin{align}
\big\|\frac{1}{|D|}\cdot\sum_{i=1}^{|D|}\rho_{f_{i}(N)R}\otimes\proj{i}_{D}-u_{M}\otimes\rho_{R}\otimes u_{D}\big\|_{1}\leq\eps\ .
\end{align}
\end{definition}

It was shown by K\"onig and Terhal that extractors with one bit output are quantum-proof~\cite[Theorem 1]{Koenig08}, and in general it is known by now that many extractor constructions are quantum-proof~\cite{Renner05,Koenig08,Koenig11,Tomamichel11,Szehr11} or suffer at most from a decent parameter loss~\cite{De12}. However, Gavinsky {\it et al.}~\cite{Gavinsky07} gave an example of a valid (though contrived) extractor that completely fails in the presence of quantum side information. Moreover there is no general understanding of when an extractor is quantum-proof. Next, we define spectral extractors and show that they are quantum-proof extractors.

\begin{definition}\label{def:class_renyi2}
A $(k,\eps)$ spectral extractor is a set of functions $\{f_{1},\ldots,f_{|D|}\}$ from $N$ to $M$ such that for the map $\psi:\ell(N)\to\ell(MD)$ with $\psi(P_{N})=\frac{1}{|D|}\cdot\sum_{i=1}^{|D|}P_{f_{i}(N)}\otimes\proj{i}_{D}$,
\begin{align}\label{eq:class_renyi2}
\lambda_{1}\big(\psi^{\dagger}\circ\psi-\tau^{\dagger}\circ\tau\big)\leq2^{k}\cdot\frac{\eps}{|M|\cdot|D|}\ ,
\end{align}
where $\lambda_{1}(\cdot)$ denotes the largest eigenvalue, and $\tau(P_{N})=\big(\sum_{j=1}^{|N|}P_{j}\big)\cdot\big(u_{M}\otimes u_{D}\big)$. For typical applications, it is sufficient to bound the second largest eigenvalue $\lambda_{2}(\psi^{\dagger}\circ\psi)$.
\end{definition}

Now, we show that spectral extractors are also quantum-proof extractors. For this we use a similar calculation as Renner {\it et al.}, who showed (directly) that families of two-universal hash functions~\cite{Renner05,Tomamichel11}, and families of pairwise independent permutations~\cite{Szehr11} give rise to quantum-proof extractors.

\begin{theorem}\label{thm:crenyi2_stability}
Every $(k,\eps)$ spectral extractor is also a quantum-proof $(k,2\sqrt{\eps})$ extractor of the same output size and the same seed size.
\end{theorem}

\begin{proof}
We write the extractor as a map $\psi:\ell(N)\to\ell(MD)$ with $\psi(P_{N})=\frac{1}{|D|}\cdot\sum_{i=1}^{|D|}P_{f_{i}(N)}\otimes\proj{i}_{D}$, and denote $\tau:\ell(N)\to\ell(MD)$ with $\tau(P_{N})=\big(\sum_{j=1}^{|N|}P_{j}\big)\cdot\big(u_{M}\otimes u_{D}\big)$. Then we get for the 1-norm
\begin{align}\label{eq:h2_1}
&\|((\psi-\tau)\otimes\cI_{R})(\rho_{NR})\|_{1}\notag\\
&=2\cdot\max_{0\leq X\leq\id}\trace\Big[\big((\psi-\tau)\otimes\cI_{R}\big)(\rho_{NR})X\Big]\notag\\
&=2\cdot\max_{0\leq X\leq\id}\trace\Big[\big((\psi-\tau)\otimes\cI_{R}\big)(\rhoh_{NR})(\id_{MD}\otimes\rho_{R}^{1/2})\notag\\
&\qquad\qquad\qquad\qquad\qquad\qquad\qquad\qquad X(\id_{MD}\otimes\rho_{R}^{1/2})\Big]\notag\\
&=2\cdot\max_{0\leq X\leq\id}\braket{((\psi-\tau)\otimes\cI_{R})(\rhoh_{NR})}{X}_{(\id\otimes\rho)}\ ,
\end{align}
where $\rhoh_{NR}=(\id_{N}\otimes\rho_{R}^{-1/2})\rho_{NR}(\id_{N}\otimes\rho_{R}^{-1/2})$, and we made use of the $(\id_{MD}\otimes\rho_{R})$-weighted Hilbert-Schmidt inner product. By using the $(2,2)$-H\"older inequality for the $(\id_{MD}\otimes\rho_{R})$-weighted Hilbert-Schmidt inner product we get
\begin{align}\label{eq:h2_2}
&2\cdot\max_{0\leq X\leq\id}\braket{((\psi-\tau)\otimes\cI_{R})(\rhoh_{NR})}{X}_{(\id\otimes\rho)}\notag\\
&\leq2\cdot\max_{0\leq X\leq\id}\|X\|_{2,(\id\otimes\rho)}\cdot\|((\psi-\tau)\otimes\cI_{R})(\rhoh_{NR})\|_{2,(\id\otimes\rho)}\ ,
\end{align}
with the $(\id_{MD}\otimes\rho_{R})$-weighted 2-norm. We estimate the first term by using the $(1,\infty)$-H\"older inequality for the $(\id_{MD}\otimes\rho_{R})$-weighted Hilbert-Schmidt inner product
\begin{align}\label{eq:h2_3}
\|X\|_{2,(\id\otimes\rho)}&=\sqrt{\braket{X}{X}_{2,(\id\otimes\rho)}}\notag\\
&\leq\sqrt{\|X\|_{\infty,(\id\otimes\rho)}\cdot\|X\|_{1,(\id\otimes\rho)}}\notag\\
&=\sqrt{\lambda_{1}(X)\cdot\trace[(\id_{MD}\otimes\rho_{R})X]}\notag\\
&\leq\sqrt{|M|\cdot|D|}\ .
\end{align}
For the second term a straightforward calculation gives
\begin{align}\label{eq:h2_4}
&\|((\psi-\tau)\otimes\cI_{R})(\rhoh_{NR})\|_{2,(\id\otimes\rho)}^{2}\notag\\
&=\braket{(\psi\otimes\cI_{R})(\rhoh_{NR})}{(\psi\otimes\cI_{R})(\rhoh_{NR})}_{(\id\otimes\rho)}\notag\\
&\quad-2\cdot\braket{(\psi\otimes\cI_{R})(\rhoh_{NR})}{(\tau\otimes\cI_{R})(\rhoh_{NR})}_{(\id\otimes\rho)}\notag\\
&\quad+\braket{(\tau\otimes\cI_{R})(\rhoh_{NR})}{(\tau\otimes\cI_{R})(\rhoh_{NR})}_{(\id\otimes\rho)}\notag\\
&=\braket{\rhoh_{NR}}{((\psi^{\dagger}\circ\psi-\tau^{\dagger}\circ\tau)\otimes\cI_{R})(\rhoh_{NR})}_{(\id\otimes\rho)}\notag\\
&=\braket{\tilde{\rho}_{NR}}{((\psi^{\dagger}\circ\psi-\tau^{\dagger}\circ\tau)\otimes\cI_{R})(\tilde{\rho}_{NR})}\ ,
\end{align}
where $\tilde{\rho}_{NR}=(\id_{N}\otimes\rho_{R}^{-1/4})\rho_{NR}(\id_{N}\otimes\rho_{R}^{-1/4})$. By~\eqref{eq:h2vshmin} the conditional min-entropy is upper bounded by the conditional R\'{e}nyi-two entropy $H_{2}(N|R)_{\rho}=-\log\|\tilde{\rho}_{NR}\|_{2}^{2}$, and this gives
\begin{align}
&\sup_{\substack{H_{\min}(N|R)_{\rho}\geq k\\\|\rho\|_{1}=1}}\braket{\tilde{\rho}_{NR}}{((\psi^{\dagger}\circ\psi-\tau^{\dagger}\circ\tau)\otimes\cI_{R})(\tilde{\rho}_{NR})}\notag\\
&\leq\frac{1}{2^{k}}\cdot\sup_{\|\tilde{\rho}\|_{2}=1}\braket{\tilde{\rho}_{NR}}{((\psi^{\dagger}\circ\psi-\tau^{\dagger}\circ\tau)\otimes\cI_{R})(\tilde{\rho}_{NR})}\notag\\
&=\frac{1}{2^{k}}\cdot\lambda_{1}(\psi^{\dagger}\circ\psi-\tau^{\dagger}\circ\tau)\ .
\end{align}
From the properties of a $(k,\eps$) spectral extractor, the claim follows.
\end{proof}

An instructive example are two-universal families of hash functions.

\begin{definition}
A set of functions $\{f_{1},\ldots,f_{|D|}\}$ from $N$ to $M$ is said to be a two-universal family of hash functions if we have for all $j\neq k\in N$,
\begin{align}\label{eq:hashing}
\frac{1}{D}\sum_{i=1}^{|D|}\delta_{f_{i}(j)=f_{i}(k)}\leq\frac{1}{|M|}\ .
\end{align}
\end{definition}

\begin{proposition}\label{prop:hashing}
A two-universal family of hash functions is a $(k,\eps)$ spectral extractor with $m=k-\log(1/\eps)$.
\end{proposition}

\begin{proof}
For any $P_{N},Q_{N}\in\ell(N)$ we calculate
\begin{align}
\braket{Q_{N}}{(\tau^{\dagger}\circ\tau)(P_{N})}=\frac{1}{|M|\cdot|D|}\cdot\sum_{j=1}^{|N|}P_{N}^{j}\sum_{k=1}^{|N|}Q_{N}^{k}\ .
\end{align}
Furthermore, we get from~\eqref{eq:hashing},
\begin{align}
&\braket{Q_{N}}{(\psi^{\dagger}\circ\psi)(P_{N})}=\frac{1}{|D|^{2}}\cdot\sum_{i=1}^{|D|}\braket{P_{N}}{(\psi_{i}^{\dagger}\circ\psi_{i})(Q_{N})}\notag\\
&=\frac{1}{|D|^{2}}\cdot\sum_{i=1}^{|D|}\sum_{j,k=1}^{|N|}\bar{P}_{N}^{j}Q_{N}^{k}\delta_{f_{i}(j)=f_{i}(k)}\notag\\
&\leq\frac{1}{|D|}\cdot\langle P_{N}|Q_{N}\rangle+\frac{1}{|M|\cdot|D|}\cdot\sum_{j=1}^{|N|}P_{N}^{j}\sum_{k=1}^{|N|}Q_{N}^{k}\ .
\end{align}
Hence, we arrive at
\begin{align}
\braket{P_{N}}{(\psi^{\dagger}\circ\psi-\tau^{\dagger}\circ\tau)(P_{N})}\leq\frac{1}{|D|}\cdot\langle X_{N}|X_{N}\rangle\ ,
\end{align}
and the claim follows.
\end{proof}

Other examples of constructions based on spectral extractors are, e.g., pairwise independent families of permutations, or constructions based on balanced expander graphs (these are weak extractors). Unfortunately, spectral extractors have the drawback of a long seed.

\begin{proposition}\label{prop:h2_longseed}
Every $(k,\eps)$ spectral extractor with input size $n$, output size $m$, and seed size $d$ necessarily has
\begin{align}
d\geq\min\{n-k,m\}+\log\frac{1}{\eps}-O(1)\ .
\end{align} 
\end{proposition}

\begin{proof}
By the same arguments as in the proof of Theorem~\ref{thm:crenyi2_stability} we have that
\begin{align}\label{eq:longseed_0}
\lambda_{1}(\psi^{\dagger}\circ\psi-\tau^{\dagger}\circ\tau)=\sup_{\|P_{N}\|_{2}^{2}\leq1}\|(\psi-\tau)(P_{N})\|_{2}^{2}\ ,
\end{align}
where $P_{N}\in\ell(N)$. Now let $Q_{N}$ be a flat $k$-source (i.e., $Q_{N}$ has $2^{k}$ non-zero entries equal to $2^{-k}$) such that the image $Q_{f_{1}(N)}$ of $Q_{N}$ under the function $f_{1}$ has support of size $|S|=\lceil2^{k}\cdot|M|/|N|\rceil$. Since $H_{2}(N)_{Q}=-\log\|Q_{N}\|_{2}^{2}\geq k$ we have with~\eqref{eq:longseed_0} that
\begin{align}
\lambda_{1}(\psi^{\dagger}\circ\psi-\tau^{\dagger}\circ\tau)&\geq\frac{2^{k}}{|D|^{2}}\cdot\sum_{i=1}^{|D|}\Big\|Q_{f_{i}(N)}-u_{M}\Big\|_{2}^{2}\notag\\
&\geq\frac{2^{k}}{|D|^{2}}\cdot\Big\|Q_{f_{1}(N)}-u_{M}\Big\|_{2}^{2}\notag\\
&\geq\frac{2^{k}}{|D|^{2}}\cdot\sum_{s=1}^{|S|}\left(Q_{f_{1}(N)}^{s}-\frac{1}{|M|}\right)^{2}\notag\\
&\geq\frac{2^{k}}{|D|^{2}}\cdot|S|\cdot\Big(\frac{1}{|S|}-\frac{1}{|M|}\Big)^{2}\ ,
\end{align}
where we used for the last inequality that the expression is maximized when $Q_{f_1(N)}$ is uniformly distributed over $S$. Hence, we get by assumption that
\begin{align}
\frac{1}{|D|}\cdot|S|\cdot\left(\frac{1}{|S|}-\frac{1}{|M|}\right)^{2}\leq\frac{\eps}{|M|}\ ,
\end{align}
or equivalently,
\begin{align}
\label{eq:lower-bound-d_0}
|D|\geq\frac{1}{\eps}\cdot\frac{\big(|M|-|S|\big)^{2}}{|M|\cdot|S|}\ .
\end{align}
Now, if $2^{k}\cdot|M|/|N|\leq1$, then $|S| = 1$ and~\eqref{eq:lower-bound-d_0} becomes 
\begin{align}
d &\geq \log \frac{1}{\eps} + 2\log(|M| - 1) - \log |M|\notag\\
&\geq m+\log\frac{1}{\eps}-2\ .
\end{align}
Otherwise, we get
\begin{align}
d 
\geq\log \frac{1}{\eps} + \log\left(\frac{N}{2^{k}}\cdot\frac{(1 - 2^k/|N| - 1/|M|)^2}{1 + |N|/(2^{k}\cdot|M|)}\right)\ .
\end{align}
But we have $\frac{|N|}{2^k\cdot|M|} \leq 1$ and thus,
\begin{align}
\frac{(1 - 2^k/|N| - 1/|M|)^2}{1 + |N|/(2^{k}\cdot|M|)} \geq \frac{1 - 2 \cdot 2^k/N}{2}\ .
\end{align}
When $k \leq n-2$, then we get
\begin{align}
d &\geq \log \frac{1}{\eps} + \log\left(\frac{N}{2^k}\cdot\frac{1}{4}\right)= n - k + \log \frac{1}{\eps} - 2\ ,
\end{align}
and otherwise the bound we aim to prove is simply implied by the general lower bound for the seed of extractors~\eqref{eq:tashma_bounds}.
\end{proof}





\section{Quantum Min-Entropy Extractors}\label{sec:quantum_min}

To understand our definition of quantum extractors, it is convenient to start with permutation based classical extractors, i.e., a family of permutations acting on the input. This family of permutations should satisfy the following property: for any probability distribution on input bit strings with high min-entropy, applying a typical permutation from the family to the input induces an almost uniform probability distribution on a prefix of the output. We define a quantum to quantum extractor in a similar way by allowing the operations performed to be general unitary transformations and the input to the extractor to be quantum.

\begin{definition}\label{def:stron_qmin}
Let $\cH_{M}\subset\cH_{N}$, $k\in[0,\log|N|]$, and $\eps>0$. A $(k,\eps)$-quantum extractor is a set of unitaries $\{U_{N}^{1},\ldots,U_{N}^{|D|}\}$ such that for all $\rho_{N}\in\cS(\cH_{N})$ with $H_{\min}(N)_{\rho}\geq k$,
\begin{align}\label{eq:qqminstrong}
\Big\|\frac{1}{|D|}\cdot\sum_{i=1}^{|D|}\trace_{N\backslash M}\big[U_{N}^{i}\rho_{N}(U_{N}^{i})^{\dagger}\big]\otimes\proj{i}_{D}-\frac{\id_{M}}{|M|}\otimes&\frac{\id_{D}}{|D|}\Big\|_{1}\notag\\
&\leq\eps\ .
\end{align}
The quantity $n=\log|N|$ is called the input size, $m=\log|M|$ the output size, and $d=\log|D|$ the seed size.\footnote{For a weak $(k,\eps)$-quantum extractor we just replace~\eqref{eq:qqminstrong} with $\Big\|\frac{1}{|D|}\cdot\sum_{i=1}^{|D|}\trace_{N\backslash M}\big[U_{N}^{i}\rho_{N}(U_{N}^{i})^{\dagger}\big]-\frac{\id_{M}}{|M|}\Big\|_{1}\leq\eps$.}
\end{definition}

We note that the seed $D$ is still classical in this definition. Alternatively, we could also define quantum extractors as general quantum channels from $\cS(\cH_{N})$ to $\cS(\cH_{M})$, and the number of Kraus operators would correspond to the dimension of the quantum seed $D$. For example, the fully depolarizing channel corresponds to a perfect extractor, independent of the min-entropy of the input. But since the minimal number of Kraus operators of the fully depolarizing channel is equal to the square of the output dimension $|M|$, it also has quantum seed size $d=2m$. However, here we restrict ourselves to quantum extractors with classical seed. It is instructive to consider extractors with domain and range consisting of qubit strings, i.e., $\cH_{N}=(\nC^{2})^{\otimes n}$ and $\cH_{M}=(\nC^{2})^{\otimes m}$, as well as with a binary seed, i.e., $D=\{0,1\}^{d}$. Examples for quantum extractors in the literature include the following:
\begin{itemize}
\item In~\cite{Dupuis10} so-called decoupling theorems were studied, and in particular it was shown that unitary 2-designs (see Definition~\ref{def:unitary2design}) are quantum extractors.
\item Ben-Aroya {\it et al.}~considered weak quantum extractors with the input size equal to the output size~\cite[Definition 5.1]{Ben-Aroya10}, and showed how to use quantum expanders for explicit constructions. See also the related work by Harrow~\cite{Harrow09_2} and references therein.
\item Hayden {\it et al.} studied quantum state randomization, which corresponds to weak $(0,\eps)$-quantum extractors with the input size equal to the output size~\cite{Hayden04}. See also the subsequent literature~\cite{Ambainis04,Dickinson06,Aubrun09}.
\end{itemize}
These constructions have many applications in quantum information theory, quantum cryptography, quantum complexity theory, and quantum physics (see, e.g., the papers above and references therein). Here, we discuss whether an extractor also works if the input is initially correlated with another quantum system. That is, we ask if an extractor is not only randomizing but decoupling as well. Note that the fully quantum conditional min-entropy can be negative for entangled states.

\begin{definition}
Let $\cH_{M}\subset\cH_{N}$, $k\in[-\log|N|,\log|N|]$, and $\eps>0$. A $(k,\eps)$-quantum extractor $\{U_{N}^{1},\ldots,U_{N}^{|D|}\}$ is decoupling if for all $\rho_{NR}\in\cS(\cH_{NR})$ with $H_{\min}(N|R)_{\rho}\geq k$,
\begin{align}
\Big\|\frac{1}{|D|}\cdot\sum_{i=1}^{|D|}\trace_{N\backslash M}\big[U_{N}^{i}\rho_{NR}&(U_{N}^{i})^{\dagger}\big]\otimes\proj{i}_{D}\notag\\
&-\frac{\id_{M}}{|M|}\otimes\rho_{R}\otimes\frac{\id_{D}}{|D|}\Big\|_{1}\leq\eps\ .
\end{align}
\end{definition}

Decoupling quantum extractors are extremely useful in quantum coding theory (see, e.g.,~\cite{Dupuis09} and references therein). In analogy to the classical case, one way of constructing quantum extractors is by means of quantum spectral extractors. In fact, all constructions (even the probabilistic ones) for quantum extractors that are known to be decoupling are based on spectral extractors.

\begin{definition}\label{def:h2_qq}
A $(k,\eps)$-quantum spectral extractor is a set of unitaries $\{U_{N}^{1},\ldots,U_{N}^{|D|}\}$ such that for the map $\psi(\rho_{N})=\frac{1}{|D|}\cdot\sum_{i=1}^{|D|}\trace_{N\ M}\big[U_{N}^{i}\rho_{N}(U_{N}^{i})^{\dagger}\big]\otimes\proj{i}_{D}$,
\begin{align}
\lambda_{1}\big(\psi^{\dagger}\circ\psi-\tau^{\dagger}\circ\tau\big)\leq2^{k}\cdot\frac{\eps}{|M|\cdot|D|}\ ,
\end{align}
where $\tau(\rho_{N})=\trace[\rho_{N}]\cdot\frac{\id_{M}}{|M|}\otimes\frac{\id_{D}}{|D|}$. For typical applications, it is sufficient to bound the second largest eigenvalue $\lambda_{2}(\psi^{\dagger}\circ\psi)$.
\end{definition}

In full analogy to the classical case, we get:

\begin{theorem}\label{thm:qrenyi2_stable}
Every $(k,\eps)$-quantum spectral extractor is also a decoupling $(k,2\sqrt{\eps})$-quantum extractor of the same output size and the same seed size.
\end{theorem}

An instructive example are unitary two-designs.

\begin{definition}\label{def:unitary2design}
A set of unitaries $\{U_1, \dots, U_L\}$ acting on $\cH$ is said to be a unitary 2-design if we have for all $M\in\cB(\cH)$ that
\begin{align}
\frac{1}{L}\cdot\sum_{i=1}^L U_i^{\otimes 2} M (U_i^{\dagger})^{\otimes 2} = \int U^{\otimes 2} M (U^{\dagger})^{\otimes 2} dU\ ,
\end{align}
where the integration is with respect to the Haar measure on the unitary group.
\end{definition}

Many efficient constructions of unitary 2-designs are known~\cite{Dankert09,Gross07}, and in an $n$-qubit space, such unitaries can typically be computed by circuits of size $O(n^{2})$.

\begin{proposition}\label{prop:unitary2design}
A unitary 2-design is a $(k,\eps)$-quantum spectral extractor with output size $m=(n+k)/2-\log(1/\sqrt{\eps})$, where $n$ denotes the input size.
\end{proposition}

Note that $k$ can be negative for entangled input states, and that the corresponding classical result for families of two-universal hash functions reads $m=k-\log\frac{1}{\eps}$ (Proposition~\ref{prop:hashing}).

\begin{proof}
For any $X_{N},Y_{N}\in\cP(\cH_{N})$ we get
\begin{align}
&\braket{X_{N}}{(\tau^{\dagger}\circ\tau)(Y_{N})}=\braket{\tau(X_{N})}{\tau(Y_{N})}\notag\\
&=\trace\left[\trace\big[X_{N}^{\dagger}\big]\cdot\frac{\id_{MD}}{|M|\cdot|D|}\cdot\trace\big[Y_{N}\big]\cdot\frac{\id_{MD}}{|M|\cdot|D|}\right]\notag\\
&=\frac{1}{|M|\cdot|D|}\cdot\trace\big[X_{N}^{\dagger}\big]\cdot\trace\big[Y_{N}\big]\ .
\end{align}
Furthermore, we calculate
\begin{align}
&\braket{X_{N}}{(\psi^{\dagger}\circ\psi)(Y_{N})}=\frac{1}{|D|^{2}}\cdot\sum_{i=1}^{|D|}\braket{X_{N}}{(\psi_{i}^{\dagger}\circ\psi_{i})(Y_{N})}\notag\\
&=\frac{1}{|D|^{2}}\cdot\sum_{i=1}^{|D|}\trace\Big[X^{\dagger}_{N}(U_{N}^{i})^{\dagger}\notag\\
&\qquad\qquad\qquad\big(\id_{N\backslash M}\otimes(\trace_{N\backslash M}[U_{N}^{i}Y_{N}(U_{N}^{i})^{\dagger}])\big)U_{N}^{i}\Big]\notag\\
&=\frac{1}{|D|^{2}}\cdot\sum_{i=1}^{|D|}\trace\Big[\big(\trace_{N\backslash M}[U_{N}^{i}X_{N}^{\dagger}(U_{N}^{i})^{\dagger}]\notag\\
&\qquad\qquad\qquad\quad\otimes\trace_{N'\backslash M'}[U_{N'}^{i}Y_{N'}(U_{N'}^{i})^{\dagger}]\big)F_{MM'}\Big]\notag\\
&=\frac{1}{|D|}\cdot\trace\Big[\big(X_{N}^{\dagger}\otimes Y_{N'}\big)\frac{1}{|D|}\sum_{i=1}^{|D|}\big((U_{N}^{i})^{\dagger}\otimes (U_{N'}^{i})^{\dagger}\big)\notag\\
&\qquad\qquad\quad(F_{MM'}\otimes\id_{NN'\backslash MM'})\big(U_{N}^{i}\otimes U_{N'}^{i}\big)\Big]\ ,
\end{align}
where we have used that the partial trace commutes with the identity, and denote the swap operator by $F_{MM'}$. Since $\{U_{N}^{1},\dots,U_{N}^{|D|}\}$ is a unitary 2-design we have that~\cite[Lemma 3.4]{Dupuis10}
\begin{align}
&\frac{1}{|D|}\sum_{i=1}^{|D|}\big((U_{N}^{i})^{\dagger}\otimes (U_{N'}^{i})^{\dagger}\big)(F_{MM'}\otimes\id_{NN'\backslash MM'})\big(U_{N}^{i}\otimes U_{N'}^{i}\big)\notag\\
&=\int\big(U_{N}^{\dagger}\otimes U_{N'}^{\dagger}\big)(F_{MM'}\otimes\id_{NN'\backslash MM'})\big(U_{N}\otimes U_{N'}\big)dU\notag\\
&=\frac{1}{|M|}\cdot\frac{|N|^{3}-|M|^{2}\cdot|N|}{|N|^{3}-|N|}\cdot\id_{NN'}\notag\\
&\quad\,+\frac{1}{|M|}\cdot\frac{|N|^{2}\cdot|M|^{2}-|N|^{2}}{|N|^{3}-|N|}\cdot F_{NN'}\ .
\end{align}
Hence, we arrive at
\begin{align}
\braket{X_{N}}{(\psi^{\dagger}\circ\psi-\tau^{\dagger}\circ\tau)(X_{N})}\leq\frac{|M|}{|N|\cdot|D|}\cdot\langle X_{N}|X_{N}\rangle\ ,
\end{align}
and the claim follows.
\end{proof}


As in the classical case, quantum spectral extractors always have a long seed.

\begin{proposition}\label{prop:qq_longseed}
Every $(k,\eps)$-quantum spectral extractor with input size $n$, output size $m$, and seed size $d$ necessarily has $d\geq\min\{n-k,m\}+\log(1/\eps)-O(1)$.
\end{proposition}

\begin{proof}
Let $\cH_{S}\subset\cH_{M}$ with $|S|=\lceil2^{k}\cdot|M|/|N|\rceil$, let $\{\ket{t}\}_{t=1}^{|N|/|M|}$ be an orthonormal basis of $\cH_{N\backslash M}$, and consider the state
\begin{align}
\gamma_{N}=\frac{|M|}{|S|\cdot|N|}\cdot\sum_{s\in S}\sum_{t=1}^{|N|/|M|}(U_{N}^{1})^{\dagger}\proj{st}_{N}U_{N}^{1}\ .
\end{align}
Since $H_{2}(N)_{\gamma}=-\log\|\gamma_{N}\|_{2}^{2}\geq k$ we have by the same arguments as in the classical case (Proposition~\ref{prop:h2_longseed}) that
\begin{align}
&\lambda_{1}(\psi^{\dagger}\circ\psi-\tau^{\dagger}\circ\tau)\notag\\
&\geq\frac{2^{k}}{|D|^{2}}\cdot\sum_{i=1}^{|D|}\Big\|\trace_{N\backslash M}\big[U_{N}^{i}\gamma_{N}(U_{N}^{i})^{\dagger}\big]-\frac{\id_{M}}{|M|}\Big\|_{2}^{2}\notag\\
&\geq\frac{2^{k}}{|D|^{2}}\cdot\Big\|\trace_{N\backslash M}\big[U_{N}^{1}\gamma_{N}(U_{N}^{1})^{\dagger}\big]-\frac{\id_{M}}{|M|}\Big\|_{2}^{2}\notag\\
&\geq\frac{2^{k}}{|D|^{2}}\cdot|S|\cdot\Big(\frac{1}{|S|}-\frac{1}{|M|}\Big)^{2}\ .
\end{align}
The rest of the proof proceeds as in the classical case (Proposition~\ref{prop:h2_longseed}), except that we use in the very end a general lower bound for the seed of quantum extractors (Proposition~\ref{prop:coneps}) instead of the corresponding bound for classical extractors~\eqref{eq:tashma_bounds}.
\end{proof}

We show in the next section that there exists a quantum extractor with seed size $O(\log(1/\eps))$ matching the simple $d \geq \log(1/\eps)$ lower bound. In contrast, any classical extractor has to satisfy $d\geq\log(n-k)+2\log(1/\eps)-O(1)$.

\begin{proposition}\label{prop:coneps}
Every $(k,\eps)$-quantum min-entropy extractor with $k\leq n-1$ ($n$ is the output size) necessarily has seed size $d\geq\log(1/\eps)$.
\end{proposition}

\begin{proof}
Let $\cH_{S}\subset\cH_{M}$ with $|S|=|M|/2$, let $\{\ket{t}\}_{t=1}^{|N|/|M|}$ be an orthonormal basis of $\cH_{N\backslash M}$, and consider the state
\begin{align}
\gamma_{N}=\frac{2}{|M|}\cdot\sum_{s\in S}\sum_{t=1}^{|N|/|M|}(U_{N}^{1})^{\dagger}\proj{st}_{N}U_{N}^{1}\ .
\end{align}
Since $H_{\min}(N)_{\sigma}=n-1$, and
\begin{align}
&\Big\|\trace_{N\backslash M}\big[U_{N}^{1}\gamma_{N}(U_{N}^{1})^{\dagger}]-\frac{\id_{M}}{|M|}\Big\|_{1}\notag\\
&=\Big\|\frac{2}{|M|}\cdot\sum_{s\in S}\proj{s}_{M}-\frac{\id_{M}}{|M|}\Big\|_{1}\notag\\
&=1\ ,
\end{align}
the claim follows.
\end{proof}


\section{Short Seeded Quantum Extractors}\label{sec:shortseeded}

Here we show that very small sets of random unitaries yield good quantum extractors.

\begin{theorem}\label{thm:shortseed}
There exists a $(k,\eps)$-quantum extractor with $m=(n+k)/2-\log(1/\eps)-O(1)$, and $d=2\log(1/\eps)+O(\log\log(1/\eps))$.
\end{theorem}

\begin{proof}
For $\rho_{N}\in\cS(\cH_{N})$ with $H_{\min}(N)_{\rho}\geq l$ we have by the extraction property of unitary 2-designs (Proposition~\ref{prop:unitary2design}) that
\begin{align}
\int\|\trace_{N\backslash M}\big[U_{N}\rho_{N}U_{N}^{\dagger}\big]-\frac{\id_{M}}{|M|}\|_{1}\;dU\leq \frac{M}{\sqrt{N\cdot2^{l}}}\ ,
\end{align}
where the integration is with respect to the Haar measure on the unitary group. This means that for each specific input there exists a unitary $U_{N}$ that extracts well. We use a measure concentration argument (based on L\'{e}vy's lemma) that gives~\cite[Theorem 3.9]{Dupuis09},
\begin{align}
\mathbf{Pr}\Big\{\Big\|\trace_{N\backslash M}\big[U_{N}\rho_{N}U_{N}^{\dagger}\big]-&\frac{\id_{M}}{|M|}\Big\|_{1}\geq \frac{M}{\sqrt{N\cdot2^{l}}}+\gamma\Big\}\notag\\
&\leq\exp(-\frac{N\gamma^{2}\cdot2^{l}}{16})\ ,
\end{align}
for $\gamma>0$. Moreover, we use a concentration of the average bound~\cite[Lemma A.2]{Fawzi11} to get
\begin{align}\label{eq:omar_concentration}
\mathbf{Pr}\Big\{\frac{1}{t}\cdot\sum_{i=1}^{t}\Big\|\trace_{N\backslash M}\big[U_{N}^{i}\rho_{N}(U_{N}^{i})^{\dagger}\big]-&\frac{\id_{M}}{|M|}\Big\|_{1}-\frac{M}{\sqrt{N\cdot2^{l}}}\geq\gamma\Big\}\notag\\
&\leq\exp(-\frac{tN\gamma^{2}\cdot2^{l}}{16})\ .
\end{align}
In order to obtain a set of unitaries that extracts well for all states, we use a net $N_{l,\delta}$ of states with $|N_{l,\delta}|\leq(5/\delta)^{2N\cdot2^{l}}$ such that for every flat $l$-source $\rho_{N}$ (i.e., $\rho_{N}$ has $2^{l}$ non-zero eigenvalues equal to $2^{-l}$) there exists $\bar{\rho}_{N}\in N_{l,\delta}$ with $\|\rho_{N}-\bar{\rho}_{N}\|_{1}\leq\delta$ for (sufficiently small) $\delta>0$~\cite[Lemma II.4]{Hayden04}. For a union bound over all $\rho_{N}\in N_{l,\delta}$,~\eqref{eq:omar_concentration} then gives
\begin{align}\label{eq:test}
\mathbf{Pr}\Big\{&\exists\rho_{N}\in N_{l,\delta}:\;\frac{1}{t}\cdot\sum_{i=1}^{t}\Big\|\trace_{N\backslash M}\big[U_{N}^{i}\rho_{N}(U_{N}^{i})^{\dagger}\big]-\frac{\id_{M}}{|M|}\Big\|_{1}\\
&-\frac{M}{\sqrt{N\cdot2^{l}}}\geq\gamma\Big\}\notag\leq\left(\frac{5}{\delta}\right)^{2N\cdot2^{l}}\cdot\exp(-\frac{tN\gamma^{2}\cdot2^{l}}{16})\ .
\end{align}
Now, we fix $\frac{M}{\sqrt{N\cdot2^{k}}}=\eps/3$ giving us $m=(n+k)/2-\log(1/\eps)-\log(3)$. Furthermore, we choose $\gamma=\eps/3$ and $t=(C/\eps^{2})\cdot\log(1/\delta)$ for some (sufficiently large) $C>0$. From~\eqref{eq:test} we then get that for all $\rho_{N}\in N_{k,\delta}$,
\begin{align}\label{eq:netstep}
\frac{1}{t}\cdot\sum_{i=1}^{t}\Big\|\trace_{N\backslash M}\big[U_{N}^{i}\rho_{N}(U_{N}^{i})^{\dagger}\big]-\frac{\id_{M}}{|M|}\Big\|_{1}\leq2\eps/3\ ,
\end{align}
with very high probability. By taking a union bound over all $l\geq k$, we get that~\eqref{eq:netstep} still holds with very high probability for all $\rho_{N}\in\bigcup_{l\geq k}N_{l,\delta}$. Hence we have shown the existence of a set of unitaries with $d=\log(t)=2\log(1/\eps)+\log\log(1/\delta)+\log(C)$ that extracts well for all $\rho_{N}\in\bigcup_{l\geq k}N_{l,\delta}$. In order to make it work for all $\rho_{N}\in\cS(\cH_{N})$ with $H_{\min}(N)_{\rho}\geq k$, we write $\rho_{N}$ as a mixture $\rho_{N}=\sum_{j}p_{j}\rho_{N}^{j}$ of flat $k$-sources $\rho_{N}^{j}$~\cite[Lemma 6.10]{Vadhan11}. For each $\rho_{N}^{j}$, we know there exists $\bar{\rho}_{N}^{j}\in\cup_{l\geq k}N_{l,\delta}$ such that $\|\rho_{N}^{j}-\bar{\rho}_{N}^{j}\|_{1}\leq\delta$. This means that for all $\rho_{N}^{j}$ we have
\begin{align}
\frac{1}{t}\cdot\sum_{i=1}^{t}\Big\|\trace_{N\backslash M}\big[U_{N}^{i}\rho_{N}^{j}(U_{N}^{i})^{\dagger}\big]-\frac{\id_{M}}{|M|}\Big\|_{1}\leq2\eps/3+\delta\ .
\end{align}
The claim follows for $\delta=\eps/3$.
\end{proof}

Thus the optimal seed size does only depend on the error $\eps$ and not on the input size $n$ and min-entropy $k$ as in the classical case. However, we do not know if the extractor from Theorem~\ref{thm:shortseed} is also decoupling. Or more generally, if any decoupling quantum extractors with seed size $d<\min\{n-k,m\}+\log(1/\eps)-O(1)$ exist (cf.~Proposition~\ref{prop:qq_longseed}).


\section{Discussion}\label{sec:discussion}

We note that our stability result for classical and quantum spectral extractors (Theorems~\ref{thm:crenyi2_stability} and~\ref{thm:qrenyi2_stable}) also works if the quantum side information is described by infinite-dimensional Hilbert spaces.


There are many open questions whose answers would have applications in quantum information theory. Concerning classical extractors we would like to gain a general understanding of when a construction is quantum-proof. Following Ta-Shma~\cite{TaShma13}, we mention that the example from~\cite{Gavinsky07} is compatible with the conjecture that every extractor is approximately quantum-proof with $\eps\mapsto\eps'=O(\eps\cdot m)$. For quantum extractors we would like to find probabilistic and explicit constructions that are decoupling but not based on spectral extractors. For quantum spectral extractors we would like to find probabilistic and explicit constructions that match our lower bound for the seed size in Proposition~\ref{prop:qq_longseed}. Finally, quantum spectral extractors $\psi$ are specified by the second largest eigenvalue $\lambda_{2}(\psi^{\dagger}\circ\psi)$, and this relates them to the study of balanced quantum expanders (as, e.g., defined in~\cite{Ben-Aroya10}).





\section*{Acknowledgment}

We acknowledge discussions with Stephanie Wehner. MB and VBS acknowledge financial support by the German Science Foundation (grant CH 843/2-1), the Swiss National Science Foundation (grants PP00P2-128455, 20CH21-138799 (CHIST-ERA project CQC)), the Swiss National Center of Competence in Research 'Quantum Science and Technology (QSIT)', the Swiss State Secretariat for Education and Research supporting COST action MP1006 and the European Research Council under the European Union's Seventh Framework Programme (FP/2007-2013) / ERC Grant Agreement no.~337603. VBS is in addition supported by an ETH Postdoctoral Fellowship. The research of OF is supported by the European Research Council grant No.~258932. OS acknowledges financial support by the Elite Network of Bavaria project QCCC


\bibliographystyle{abbrv}
\bibliography{library}


\end{document}